\documentclass{article}[12pt]
\usepackage{amsmath,amssymb,amsthm,color,url,framed}
\usepackage{graphicx,enumerate}
\usepackage[margin=2cm]{geometry}

\newtheorem{theorem}{Theorem}
\newtheorem{remark}[theorem]{Remark}
\newtheorem{proposition}[theorem]{Proposition}
\newtheorem{definition}[theorem]{Definition}
\newtheorem{corollary}[theorem]{Corollary}
\newtheorem{lemma}[theorem]{Lemma}
\newcommand{\bfi}{\bfseries\itshape}
\def\MM#1{\boldsymbol{#1}}

\def\contract{\makebox[1.2em][c]{\mbox{\rule{.6em}
{.01truein}\rule{.01truein}{.6em}}}}

\numberwithin{equation}{section}

\title{Variational Principles for Stochastic Fluid Dynamics}
\author{Darryl D Holm
\\Mathematics Department
\\Imperial College London}
\date{ \textit{Proc Roy Soc A} RSPA-2014-0963 \\ $\,$\\ \small
Keywords: Geometric mechanics, stochastic fluid models; cylindrical stochastic processes; \\
multiscale fluid dynamics; symmetry reduced variational principles}                                           

\begin{document}
\maketitle
\makeatother
\begin{abstract}
This paper derives stochastic partial differential equations (SPDEs) for fluid dynamics from a stochastic variational principle (SVP). The Legendre transform of the Lagrangian formulation of these SPDEs yields their Lie-Poisson Hamiltonian form. The paper proceeds by: taking variations in the SVP to derive stochastic Stratonovich fluid equations; writing their It\^o representation; and then investigating the properties of these stochastic fluid models in comparison with each other, and with the corresponding deterministic fluid models. The circulation properties of the stochastic Stratonovich fluid equations are found to closely mimic those of the deterministic ideal fluid models. As with deterministic ideal flows, motion along the stochastic Stratonovich paths also preserves the helicity of the vortex field lines in incompressible stochastic flows. However, these Stratonovich properties are not apparent in the equivalent It\^o representation, because they are disguised by the quadratic covariation drift term arising in the Stratonovich to It\^o transformation. This term is a geometric generalisation of the quadratic covariation drift term already found for scalar densities in Stratonovich's famous 1966 paper. The paper also derives motion equations for two examples of stochastic geophysical fluid dynamics (SGFD); namely, the Euler-Boussinesq and quasigeostropic approximations.



\end{abstract}

%
\newpage

\section{Introduction}

In this paper we propose an approach for including stochastic processes as cylindrical noise in systems of evolutionary partial differential equations (PDEs) that derive from variational principles which are invariant under a Lie group action. Such dynamical systems are called Euler-Poincar\'e equations \cite{MaRa1994,HoMaRa1998}. The main objective of the paper is the inclusion of stochastic processes in ideal fluid dynamics, in which case the variational principle is invariant under the Lie group of smooth invertible maps acting to relabel the reference configuration of Lagrangian coordinates for the fluid. Examples include Euler's fluid equations for incompressible flows and also geophysical fluid dynamics (GFD) of ocean and atmosphere circulation. The approach is via a stochastic extension of the well known variational derivation of the Eulerian representation of ideal fluid dynamics \cite{HoMaRa1998}. 

The resulting stochastic partial differential equations (SPDEs) contain a type of multiplicative, cylindrical, Stratonovich noise that depends on the \emph{gradients} of the solution variables. This unfamiliar feature does not interfere with the passage to the It\^o representation, though, since the space variable is treated merely as a parameter when dealing with cylindrical noise. That is, one may regard the cylindrical noise process as a finite dimensional stochastic process parametrized by $\MM x$ (the space variable).  Then, the Stratonovich equation makes analytical sense pointwise, for each fixed $\MM x$.  Once this is agreed, then the transformation to It\^o by the standard method also makes sense pointwise in space. 

To specify the spatial correlations required in applications of the cylindrical Stratonovich noise process, we advocate the strategy of applying Proper Orthogonal Decompositions (PODs) to the appropriate numerical and observational data available at resolvable scales \cite{BeHoLu1993}. One may then regard the stochastic process as arising physically as the effect of sub-grid scale degrees of freedom on the resolved scales of fluid motion.%
\footnote{PODs are also called Empirical Orthogonal Functions (EOFs), Karhunen-Loeve projections (KLPs), and Singular Value Decompositions (SVDs) and they comprise a standard approach that has been useful in a variety of contexts \cite{BeHoLu1993}.}

In more detail, the aim of this paper is to use the methods of geometric mechanics to enable fluid dynamics to be effectively adapted to include Stratonovich stochastic processes. 
Toward this end, we derive a set of stochastic fluid equations for the motion of an either compressible, or incompressible fluid in $\mathbb{R}^3$ from a \emph{stochastically constrained} variational principle $\delta S = 0$, with action, $S$, given by
\begin{align}
S(u,p,q)  &=
\int \bigg( \ell(u,q) dt 
+ \left\langle  p\,,\,{dq} + \pounds_{dx_t} q\,\right\rangle_V \bigg)
\,,\label{SVP1}
\end{align}
where $\ell(u,q)$ is the unperturbed deterministic fluid Lagrangian, written as a functional of velocity vector field $u$ and advected quantities $q$.
Here, $u\in\mathfrak{X}(\mathbb{R}^3)$ is the fluid velocity vector field, the angle brackets  
\begin{equation}
\langle \,p\,,\,q\,\rangle_V:=\int <p(x),q(x,t)>dx
\label{L2pairing}
\end{equation}
denote the spatial $L^2$ integral  over the domain of flow of the pairing $<p\,,\,q>$ between elements $q\in V$ and their dual elements $p\in V^*$. In \eqref{SVP1}, the quantity $p\in V^*$ is a Lagrange multiplier and $\pounds_{dx_t}q$ is the \emph{Lie derivative} of an {advected quantity} $q\in V$, along a vector field $dx_t$ defined by the following sum of a drift velocity $u(x,t)$ and Stratonovich stochastic process with \emph{cylindrical noise} parameterised by spatial position $x$, \cite{Pa2007,Sc1988} 
\begin{equation}
dx_t(x) = u(x,t)\,dt - \sum_i \xi_i(x)\circ dW_i(t)
\,.\label{vel-vdt}
\end{equation}
\begin{remark}\rm
Notice at the outset that $dx_t(x)$ is a stochastic vector field parameterised by the spatial position $x$. The corresponding time integral of this vector field is a stochastic \emph{function} $(\omega,t,x)\to x_t(\omega,x)$ defined as%
\footnote{The symbol $\omega$ for stochastic quantities will be understood throughout, but will not be written explicitly hereafter.}
\begin{equation}
x_t(x) = x_0(x) + \int_0^t u(x,s)\,ds - \sum_i\int_0^t\xi_i(x)\circ dW_i(s)
\,.\label{vel-path}
\end{equation}
In particular, the $x_t$ treated here is \emph{not} the stochastic flow satisfying the SDE 
\begin{equation}
x_t(x) \ne x + \int_0^t u(x_s(x),s)\,ds - \sum_i\int_0^t\xi_i(x_s(x))\circ dW_i(s)
\,,\label{vel-SDE}
\end{equation}
which arises in other treatments of stochastic fluid dynamics, see, e.g., \cite{ArChCr2012}.
This difference in the definitions between the stochastic vector fields treated here and the usual notation in the SDE setting for stochastic fluids is necessary for our purposes here, and it greatly facilitates the analysis. We want to combine geometric mechanics with stochastic analysis for fluids. However, the usual flow relationships between Lagrangian particle maps and Eulerian fluid velocities are problematic in the stochastic setting, since expressions involving tangent vectors to stochastic processes are meaningless. Therefore, we shall develop a theory \emph{entirely} within the Eulerian interpretation of fluid dynamics. That is, all the variables discussed below will depend parametrically on the spatial coordinate $x$. 
\end{remark}

At this point, we have introduced Stratonovich stochasticity into the action principle for fluids in \eqref{SVP1} through the constraint that advected quantities $q$ should evolve by following the Stratonovich perturbed vector field. This advection law is formulated as a Lie derivative with respect to the Stratonovich stochastic vector field in \eqref{vel-vdt}.
For mathematical discussions of Lie derivatives with respect to stochastic vector fields, see \cite{ChCr2013, CiCr1999, Dr1999, IkWa1981, KaSh1994}. 

The definition of Lie derivative we shall use here is the standard Cartan definition in terms of the action of the differential operator ${\rm d}$ acting on functions of the spatial coordinate $x$,
\begin{equation}
\pounds_{dx_t} q = {\rm d}(i_{dx_t}q) + i_{dx_t}{\rm d}q
\,.\label{LieDeriv-def}
\end{equation}
where $i_{dx_t}{\rm d}q$ denotes insertion of the vector field $dx_t(x)$ into the differential form ${\rm d}q(x)$, with the usual rules for exterior calculus. For a review of exterior calculus in the context of fluid dynamics, see, e.g., \cite{Ho2011}. 

One may interpret $dx_t(x)$ in \eqref{vel-vdt} as the decomposition of a vector field defined at position $x$ and time $t$ into a time-dependent drift velocity $u(x,t)$ and a stochastic vector field. The time-independent quantities $\xi_i(x)$ with $i=1,2,\dots,K$ in the cylindrical stochastic process are usually interpreted as ``diffusivities" of the  stochastic vector field, and the choice of these quantities must somehow be specified from the physics of the problem to be considered. 
Here, we will specify the diffusivities $\xi_i(x)$ as a set of preassigned physical spatial correlations for the stochasticity. This information is to be provided during the formulation of the problem under consideration. As an example, we will  interpret the $\xi_i(x)$ with $i=1,2,\dots,K$ as spatial correlations obtained from, say, coarse-grained observations or computations (e.g., as PODs) which supply the needed information for $K$ independent Wiener (Brownian) processes, $dW_i(t)$, in the Stratonovich sense. Note that the number of vector fields $K$ need not be equal to the number of spatial dimensions. 

The $L^2$ pairing $\left\langle\,\cdot\,,\,\cdot\,\right\rangle_V$ in the Stochastic Variational Principle (SVP) written in \eqref{SVP1} with  Lagrange multiplier $p\in T^*V$ enforces the \emph{advection condition} that the quantity $q\in V$ is preserved along the Stratonovich stochastic path.
\eqref{vel-vdt}, namely,
\begin{equation}
{dq} + \pounds_{dx_t} q = 0
\,.\label{advec-cond}
\end{equation}
The advection relation \eqref{advec-cond} for the quantities $q\in V$ may be regarded as a \emph{stochastic constraint} imposed on the variational principle \eqref{SVP1} via the Lagrange multiplier $p$. Requiring that the solution to \eqref{advec-cond} exists locally in time amounts to assuming that the ``back-to-labels'' map for the solution of \eqref{vel-vdt} exists locally in time for the flow generated by the vector field \eqref{vel-vdt}, cf. \cite{CoIy2008, Ey2010}. 

An interesting physical situation occurs when numerical and observational data are available for comparison with the dynamics of the stochastic model. 
This is the situation in which the standard method of Proper Orthogonal Decompositions (PODs) may become useful. In particular, one may compute the PODs corresponding to the dominant correlations in this numerical and observational data, then select among those PODs according to well-defined physical criteria the modes corresponding to $\xi_i(x)$ with $i=1,2,\dots,K$ for finite $K$. In that case, the stochastic terms in equation \eqref{vel-vdt} represent a specific finite set of spatially-correlated, but unresolved, degrees of freedom represented by noise, that are coupled nonlinearly to the deterministic evolution of velocity $u$ and advected quantity $q$ through the Lagrange-to-Euler stochastic tangent map in equation \eqref{vel-vdt}. For example, one may choose the PODs $\xi_i(x)$ with $i=1,2,\dots,K$ as the first $K$ eigenfunctions of the correlation tensor for the observed or simulated velocity fields determined by \cite{BeHoLu1993}
\begin{align}
\int \big\langle u(x) u(x') \big\rangle \xi_i(x')\,dx' = \lambda_i^2 \xi_i(x)
\quad\hbox{no sum on }i=1,2,\dots,K
\quad\hbox{with }\lambda_1>\lambda_2>\dots>\lambda_K
\,,
\label{PODs-def}
\end{align}
so that these first $K$ eigenfunctions represent the highest correlations of velocity and thus account for the greatest fraction of kinetic energy in the data, compared to any other set of the same dimension. (The eigenvalues $\lambda_i^2$ are all positive, because of their interpretation as the relative kinetic energies of the eigenfunctions $\xi_i(x)$.)   Naturally, one might want to weight the eigenfunctions in equation \eqref{vel-vdt} as $\xi_i(x)\to \lambda_i \xi_i(x)$ (the positive square roots $ \lambda_i$ of their relative kinetic energy eigenvalues  $\lambda_i^2$), to give them the correct dimensional meaning (velocity) and relative importance.

\begin{remark}\rm
Of course, the strategy of dividing the solution into essential and nonessential modes, then replacing the dynamical effects of the nonessential modes by noise is well known. For example, one may refer to Majda and Timofeyev \cite{MaTiVa2003} for the history and motivation for this approach, as well as descriptions of their own approach, called Stochastic Mode Reduction (SMR). Some versions of SMR can be based, for example, on truncating the Fourier representation of the numerical method at a certain level that defines the essential modes, then stochastically modelling the non-linear interactions among the remaining modes, deemed to be nonessential.
\end{remark}
\paragraph{Main results.} The SPDEs which will result from the stochastically constrained variational principle $\delta S = 0$ for $S$ defined in \eqref{SVP1} are expressed in Stratonovich form in terms of the Lie-derivative operation $\pounds_{dx_t}$ as 
\begin{align}
d \frac{\delta  \ell}{\delta  u} 
+ \pounds_{dx_t} \frac{\delta  \ell}{\delta  u}
- \frac{\delta  \ell}{\delta  q}\diamond q\,dt
= 
0
\,,\quad\hbox{and}\quad
dq + \pounds_{dx_t}q 
= 0
\,,
\label{FSPDEs-Stratonovich}
\end{align}
in which $dx_t$ is the Eulerian vector field in equation \eqref{vel-vdt} for the velocity along the Lagrangian Stratonovich stochastic path 
and the diamond operation $(\diamond)$ will be explained below in Definition \ref{diamond-def}. 

The corresponding It\^o forms of these equations are
\begin{align}
\begin{split}
d \frac{\delta  \ell}{\delta  u}  
+ \pounds_{d\widehat{x}_t} \frac{\delta  \ell}{\delta  u} 
- \frac{\delta \ell}{\delta q}\diamond q\,dt 
&=
\frac12 \sum_{i,j} \pounds_{\xi_j(x)}
\left(\pounds_{\xi_i(x)} \frac{\delta  \ell}{\delta  u} \right)  \,[dW_i(t),dW_j(t)]
\,,\\
dq + \pounds_{d\widehat{x}_t}q 
&= 
\frac12 \sum_{i,j} \pounds_{\xi_j(x)}(\pounds_{\xi_i(x)}q)  \,[dW_i(t),dW_j(t)]
\,,
\end{split}
\label{FSPDEs-Ito}
\end{align}
in which the Eulerian vector field $d\widehat{x}_t$ tangent to the Lagrangian It\^o stochastic path $\widehat{x}_t$ is given by
\begin{align}
d\widehat{x}_t
= u(x,t)\,dt - \sum_i \xi_i(x) \,dW_i(t) 
\,.\label{vel-ydt}
\end{align}
In equation \eqref{FSPDEs-Ito}, the quantities $[dW_i(t),dW_j(t)]$ with $i,j=1,2,\dots,K$ for $K$ independent stochastic processes denote the quadratic covariations of the temporal It\^o noise. For Brownian processes, these quantities satisfy $[dW_i(t),dW_j(t)]=0$ unless $i=j$ and satisfy $[dW_j(t),dW_j(t)]=dt$ (no sum) \cite{Sh1996}. Hereafter, in choosing Brownian processes, we may write $[dW_i(t),dW_j(t)]=\delta_{ij}dt$.

\paragraph{Interpretations of the main results as Kelvin circulation theorems for incompressible flow.} The interpretations of the equations \eqref{FSPDEs-Stratonovich}--\eqref{FSPDEs-Ito} may be expressed quite succinctly in the case of incompressible flow. In that case, volume elements are preserved and the advected variables $q$ are absent. The corresponding Kelvin circulation theorems about the evolution of the integral of the \emph{circulation 1-form} $(\delta  \ell/\delta  u)$ around a material loop are, as follows.

\noindent
For the Stratonovich case,
\begin{align}
d \oint_{c(t)} \frac{\delta  \ell}{\delta  u} = 0
\quad\hbox{for loops $c(t)$ governed by}\quad dc(t) = -\,\pounds_{dx_t}c(t)\,,
\label{Kel-Strat-intro}
\end{align}
and equivalently  for the It\^o case,
\begin{align}
d \oint_{\widehat{c}(t)} \frac{\delta  \ell}{\delta  u} = 
\oint_{\widehat{c}(t)}
\frac12 \sum_{i,j} \pounds_{\xi_j(x)}
\left(\pounds_{\xi_i(x)} \frac{\delta  \ell}{\delta  u} \right)  \,\delta_{ij}\,dt
\quad\hbox{for}\quad d\widehat{c}(t) = -\,\pounds_{d\widehat{x}_t}\widehat{c}(t)\,.
\label{Kel-Ito-intro}
\end{align}
Thus, the Kelvin theorem in \eqref{Kel-Strat-intro} shows that circulation is \emph{conserved} for loops moving along the Stratonovich stochastic path with velocity vector field $dx_t$ in equation \eqref{vel-vdt}. However, perhaps not surprisingly, because the velocities of the loops are different, the equivalent Kelvin theorem \eqref{Kel-Ito-intro} shows that the Stratonovich circulation law is \emph{masked} in the It\^o formulation for loops moving along the  It\^o stochastic path  with velocity $d\widehat{x}_t$ in equation \eqref{vel-ydt}, because the It\^o terms cannot be expressed as a single Lie derivative of the circulation 1-form. Specifically, the circulation created in these It\^o loops in \eqref{Kel-Ito-intro} is determined from the spatial correlation vector fields $\xi_i(x)$.

The proofs of these Kelvin circulation theorems are straightforward. For example, equation \eqref{Kel-Ito-intro} is proven, as follows,
\begin{align}
\frac{d}{dt} \oint_{\widehat{c}(t)} \frac{\delta  \ell}{\delta  u} 
= 
\oint_{\widehat{c}(t)} 
\big(\partial_t + \pounds_{d\widehat{x}_t}\big) \frac{\delta  \ell}{\delta  u} 
=
 \frac12
 \oint_{\widehat{c}(t)} 
\sum_{j} \pounds_{\xi_j(x)}
\left(\pounds_{\xi_j(x)} \frac{\delta  \ell}{\delta  u} \right)
\,,
\label{Kel-xIto-intro-proof}
\end{align}
in which the last step is made by referring to equation \eqref{FSPDEs-Ito} for the case when the advected quantities $q$ are absent.

\begin{remark}\rm
When the noise is completely uncorrelated in each spatial dimension so that $\xi_i=const$, for $i=1,2,3$, and also provided $q=1$ (volume preservation), the double Lie-derivative operator $\sum_i \pounds_{\xi_i}(\pounds_{\xi_i}\,\cdot\,)$ appearing in equations \eqref{FSPDEs-Ito},  \eqref{Kel-Ito-intro} and \eqref{Kel-xIto-intro-proof} reduces to the metric Laplacian operator, $\Delta={\rm div \, grad}$. Some foundational results on related SPDEs can be found in Flandoli \cite{Fl2011,FlMaNe2014} and references therein.
\end{remark}

\begin{remark}\rm
It remains to define the diamond operation $(\diamond)$ appearing in equations \eqref{FSPDEs-Stratonovich} and \eqref{FSPDEs-Ito}. The diamond operation $(\diamond)$ appears when we include potential energy terms depending on the advected variables $q$ in the stochastically constrained variational principle in \eqref{SVP1}.
\end{remark}

\begin{definition}[The diamond operation]\label{diamond-def}\rm
On a manifold $M$, the diamond operation $(\diamond): T^*V\to \mathfrak{X}^*$ is defined for a vector space $V$ with $(q,p)\in T^*V$ and vector field $\xi\in \mathfrak{X}$ is given in terms of the Lie-derivative operation $\pounds_u$ by 
\begin{equation}
\left\langle  p\diamond q\,,\,\xi \,\right\rangle_\mathfrak{X}
:=
\left\langle  p\,,\,- \pounds_\xi q\,\right\rangle_V 
\label{diamond-def-}
\end{equation}
for the pairings $\langle  \,\cdot\,,\,\cdot\,\rangle_V: T^*V\times TV\to \mathbb{R}$ and
$\langle  \,\cdot\,,\,\cdot\,\rangle_\mathfrak{X}: \mathfrak{X}^*\times \mathfrak{X}\to \mathbb{R}$
with $p\diamond q\in \mathfrak{X}^*$. 
\end{definition}

\begin{remark}[Momentum map]\label{momap-def}\rm
The quantity $J(q,p)=p\diamond q$ in \eqref{diamond-def-} defines the \emph{cotangent-lift momentum map} for the action of the vector fields $\xi\in \mathfrak{X}$ on the vector space $V$ \cite{HoMaRa1998}. In terms of the momentum map $p\diamond q$, the action integral $S$ in \eqref{SVP1} for the stochastic variational principle $\delta S = 0$ may be written \emph{equivalently} as
\begin{align}
S(u,p,q)  &=
\underbrace{
\int \bigg( \ell(u,q) + \left\langle  p\,,\,\frac{dq}{dt} + \pounds_u q\,\right\rangle_V \bigg)dt
}_{\hbox{Lebesgue integral}}
+
\underbrace{
\int \sum_{i}  \left\langle p\diamond q \,,\,\xi_i(x)
\right\rangle_{\mathfrak{X}}\circ dW_i(t)
}_{\hbox{Stratonovich integral}}
\,.\label{SVP2}
\end{align}
Thus, the vector-field coupling between the deterministic and stochastic parts of the SVP  in this equivalent form in \eqref{SVP2} of the action integral \eqref{SVP1} is through the momentum map in \eqref{diamond-def-}. In finite dimensions, formula \eqref{SVP2} fits within the framework of \cite{Bi1981,LaCa-Or2008} for stochastic canonical Hamilton equations and generalises the work of \cite{BR-O2009} for stochastic variational integrators applied to the rigid body. This sort of momentum-map coupling in finite dimensions was also noted for symmetry reduction of mechanical systems with stochastic nonholonomic constraints studied in \cite{HoRa2012}. 
\end{remark}

\begin{remark}[Variational derivations of the Navier--Stokes equations from stochastic equations] \rm
The derivation of the Navier--Stokes equations in the context of stochastic processes has a long and well-known history. See. e.g., Constantin and Iyer \cite{CoIy2008}, Eyink \cite{Ey2010}, and references therein. Previous specifically variational treatments of stochastic fluid equations generally started from the famous remark by Arnold [1966] (about Euler's equations for the incompressible flow of an ideal fluid being geodesic for kinetic energy given by the $L^2$ norm of fluid velocity) and they have mainly treated It\^o noise in this context. For more discussion of these variational derivations of stochastic fluid equations and their relation to the Navier-Stokes equations, one should consult original sources such as, in chronological order, Inoue and Funaki \cite{InFu1979}, Rapoport \cite{Ra2000,Ra2002}, Gomes \cite{Go2005}, Cipriano and Cruzeiro \cite{CiCr2007}, Constantin and Iyer \cite{CoIy2008}, Eyink \cite{Ey2010}, Gliklikh \cite{Gl2010}, Arnaudon, Chen and Cruzeiro \cite{ArChCr2012}. We emphasise that the goal of the present work is to derive SPDEs for fluid dynamics by following the stochastic variational strategy outlined above.  It is not our intention to derive the Navier--Stokes equations in the present context. However, as mentioned previously, in imposing the stochastic constraint \eqref{advec-cond}, we have assumed the existence of a ``back-to-labels'' map. This assumption is also often made in the derivation of the the Navier--Stokes equations in a stochastic setting; see e.g., \cite{CoIy2008, Ey2010}.  For additional information, review and background references for random perturbations of PDEs and fluid dynamic models, viewed from complementary viewpoints to the present paper, see also Flandoli et al. \cite{Fl2011,FlMaNe2014}. In particular, Flandoli et al. \cite{Fl2011,FlMaNe2014} studies the interesting possibility that adding stochasticity can have a regularising effect on fluid equations which might otherwise be ill-posed. 
\end{remark}

\paragraph{Plan of the paper.} 
The remainder of the paper will derive the SPDEs for fluids in \eqref{FSPDEs-Stratonovich} and \eqref{FSPDEs-Ito} from the variational principle $\delta S=0$ with stochastic action integral $S$ given in \eqref{SVP1}, or equivalently  \eqref{SVP2}. Toward this objective, we shall take the following steps. 

Section \ref{StratVP-sec} will derive the Stratonovich motion equations for $S$ in  \eqref{SVP1} in both the Lagrangian and Hamiltonian formulations. 
Section \ref{Strat-Ito-subsec} will write the It\^o representation of these equations. 
Section \ref{KelThm-sec} will treat the abstract Kelvin circulation theorems in their Stratonovich \ref{KelThmStrat-subsec} and It\^o forms \ref{KelThmIto-subsec}.
Section \ref{Examples} will consider examples of stochastic fluid flows in the Stratonovich representation. Section \ref{SEP-incompress-subsec} will discuss Stratonovich stochastic fluid flows without advected quantities. In particular, Section \ref{SEP-incompress-subsec} will derive the Kelvin circulation theorem for Stratonovich stochastic fluids and verify their preservation of helicity (the linkage number for the vorticity field lines). In addition, Section \ref{SEP-incompress-subsec} will introduce Stratonovich stochasticity into the Euler-alpha model from turbulence theory. Section \ref{Diamond-operator} will derive the contributions of various advected quantities to the motion equation.  Section \ref{GFD-subsec} will treat the effects of these advected quantities in two examples of stochastic geophysical fluid dynamics (GFD). These two stochastic GFD examples comprise the stochastic Euler-Boussinesq equations and the stochastic quasigeostrophic (SQG) equations. In all of these examples, we will present both the Stratonovich and It\^o forms of the stochastic fluid equations and contrast their implications.
Section \ref{conclus-sec} will provide conclusions, further discussion, and outlook.

We will proceed formally without addressing technical issues of stochastic analysis, by assuming all the objects we introduce in the paper are semimartingales. This assumption is possible, because the parametric spatial dependence of the dynamical variables allows essentially finite-dimensional stochastic methods to be applied at each point of space. From that viewpoint, the paper presents a slight generalisation of earlier work by Bismut \cite{Bi1981}, L\'azaro-Cam\'i and Ortega \cite{LaCa-Or2008}, and Bou-Rabee and Owhadi \cite{BR-O2009}, which unifies their Hamiltonian and Lagrangian approaches to temporal stochastic dynamics, and extends them to the case of cylindrical noise in which spatial dependence is parametric, while temporal dependence is stochastic. Having made this assumption, we will be able to apply the normal rules of variational calculus to the Stratonovich integrals in Hamilton's principle \eqref{SVP1} to derive the equations of motion. We will then transform the equations to the It\^o side and derive the expected drift terms as generalizations of the second-order operators which first appeared in the original paper of Stratonovich \cite{St1966}. 
A general principle will given in Theorem \ref{SEP-thm} and results for a series of examples will be  derived. The present approach introduces new stochastic terms into these examples which improve the geometric structure of the equations and preserve the invariants of the underlying deterministic models. These new stochastic terms contain multiplicative noise depending on the gradients of the solution variables. 
We hope this paper will stimulate new research, not only by experts in geometric mechanics, but also by those who approach stochastic fluid dynamics using more analytical methods, because the equations we derive here are new and pose new challenges.


\section{Stratonovich stochastic variational principle}\label{StratVP-sec}

\subsection{Stochastic Euler-Poincar\'e (SEP) formulation}

\begin{theorem}[Stratonovich Stochastic Euler-Poincar\'e equations]\label{SEP-thm}\rm$\,$
The action for the stochastic variational principle $\delta S = 0$ in \eqref{SVP2}, 
\begin{align}
S(u,p,q)  &=
\underbrace{
\int \bigg( \ell(u,q) + \left\langle  p\,,\,\frac{dq}{dt} + \pounds_u q\,\right\rangle_V \bigg)dt
}_{\hbox{Lebesgue integral}}
+
\underbrace{
\int \sum_{i}  \left\langle p\diamond q \,,\,\xi_i(x)
\right\rangle_{\mathfrak{X}}\circ dW_i(t)
}_{\hbox{Stratonovich integral}}
\,,\label{SVP2-redux}
\end{align}
leads to the following Stratonovich form of the \emph{stochastic} Euler--Poincar\'e (SEP) equations
\begin{align}
dm + \pounds_{dx_t}m - \frac{\delta \ell}{\delta q}\diamond q\,dt 
=
0
\,,
\qquad
dq = -\,\pounds_{dx_t}q
\,,
\qquad
dp = \frac{\delta \ell}{\delta q}\,dt  + \pounds_{dx_t}^Tp
\,,
\label{SEP-eqns-thm}
\end{align}
where $dx_t \in \mathfrak{X}$ is the Stratonovich stochastic vector field  in equation \eqref{vel-vdt} and  \begin{align}
m:=\frac{\delta \ell}{\delta u} = p\diamond q \in \mathfrak{X}^*
\label{m-def-thm}
\end{align}
is the 1-form density of momentum. 
\end{theorem}
\begin{proof}
The first step is to take the elementary variations of the action integral \eqref{SVP2-redux}, to find
\begin{align}
\delta u:\quad
\frac{\delta \ell}{\delta u} - p\diamond q = 0
\,,\qquad
\delta p:\quad
dq + \pounds_{dx_t} q   = 0
\,,\qquad
\delta q:\quad
\frac{\delta  \ell}{\delta  q}dt - dp 
+ \pounds_{dx_t}^Tp  = 0
\,.
\label{var-eqns-thm}
\end{align}
The first variational equation captures the relation \eqref{m-def-thm}, and the latter two equations in \eqref{var-eqns-thm} produce the corresponding equations in \eqref{SEP-eqns-thm}. The governing equation for $m$ in \eqref{SEP-eqns-thm} will be recovered by using the result of the following Lemma.
\begin{lemma}\rm\label{Lemma-m-eqn}
Together, the three equations in \eqref{var-eqns-thm} imply the first formula in \eqref{SEP-eqns-thm}, namely
\begin{align}
dm - \frac{\delta \ell}{\delta q} \diamond q \,dt = - \pounds_{dx_t} m
\,.\label{m-eqn-lem}
\end{align}
\begin{proof}
For an arbitrary $\eta\in \mathfrak{X}$, one computes the pairing
\begin{align}
\begin{split}
\left\langle 
dm - \frac{\delta \ell}{\delta q} \diamond q \,dt  \,,\, \eta 
\right\rangle_{\mathfrak{X}}
&=   
\left\langle 
- \frac{\delta \ell}{\delta q} \diamond q  + dp\diamond q + p\diamond dq\,,\, \eta 
\right\rangle_{\mathfrak{X}}
\\
\hbox{By equation \eqref{var-eqns-thm} } &=   
\left\langle 
(\pounds_{dx_t}^Tp) \diamond q 
- p\diamond \pounds_{dx_t} q\,,\, \eta 
\right\rangle_{\mathfrak{X}}
\\&=   
\left\langle 
p\,,\, (-\pounds_{dx_t} \pounds_{\eta} + \pounds_{\eta} \pounds_{dx_t} )q\,
\right\rangle_{V}
\\&=   
\left\langle 
p\,,\, {\rm ad}_{dx_t}{\eta}\,q\,
\right\rangle_{V}
=
-\left\langle 
p\diamond q\,,\, {\rm ad}_{dx_t}{\eta}\,
\right\rangle_{\mathfrak{X}}
\\&=
-\left\langle 
 {\rm ad}^*_{dx_t}(p\diamond q)\,,\,{\eta}\,
\right\rangle_{\mathfrak{X}}
=
-\,\Big\langle 
 \pounds_{dx_t}m\,,\,{\eta}\,
\Big\rangle_{\mathfrak{X}}\,.
\end{split}
\label{calc-lem}
\end{align}
Since $\eta\in \mathfrak{X}$ was arbitrary, the last line completes the proof of the Lemma. In the last step we have used the fact that coadjoint action is identical to Lie-derivative action for vector fields acting on 1-form densities. 
\end{proof}
\end{lemma}
In turn, the result of Lemma \ref{Lemma-m-eqn} now produces the $m$-equation in \eqref{SEP-eqns-thm} of Theorem \ref{SEP-thm}.
\end{proof}
 
\section{Stratonovich $\to$ It\^o equations}\label{Strat-Ito-mean-sec}

\subsection{Stratonovich form}

Since $dx_t = udt - \sum_i \xi_i(x)\circ dW_i(t)$, one may rewrite the Stratonovich Euler-Poincar\'e (SEP) equations \eqref{SEP-eqns-thm} in Theorem \ref{SEP-thm}, so as to separate out the various Lie derivative operations, as follows,
\begin{align}
\begin{split}
&dm + \pounds_{u}m\,dt - \frac{\delta \ell}{\delta q}\diamond q\,dt 
=
\sum_i \pounds_{\xi_i(x)}m \circ dW_i(t)
\,,\\&
dq + \pounds_{u}q\,dt = \sum_i \pounds_{\xi_i(x)}q \circ dW_i(t)
\,,\\&
dp - \pounds_{u}^Tp\,dt - \frac{\delta \ell}{\delta q}\,dt  
= -\sum_i \pounds_{\xi_i(x)}^Tp
\circ dW_i(t)
\,,\end{split}
\label{SEP-eqns-thm-again}
\end{align}
in which the stochastic process terms all appear on the right-hand sides of the equations.  When those terms are absent, one recovers standard deterministic ideal fluid dynamics. 

\subsection{It\^o form}\label{Strat-Ito-subsec}

The corresponding It\^o forms of the equations in \eqref{SEP-eqns-thm-again} are found by using It\^o's formula to identify the quadratic covariation terms as
\begin{align}
\begin{split}
dm + \pounds_{d\widehat{x}_t}m\,dt - \frac{\delta \ell}{\delta q}\diamond q\,dt 
&=
\frac12
\sum_{j} \pounds_{\xi_j(x)}\left(\pounds_{\xi_j(x)}m\right) \,dt
\,,\\
dq + \pounds_{d\widehat{x}_t}q\,dt 
&= 
\frac12
\sum_{j} \pounds_{\xi_j(x)}\left(\pounds_{\xi_j(x)}q\right)  \,dt
\,,\\
dp - \pounds_{d\widehat{x}_t}^Tp\,dt - \frac{\delta \ell}{\delta q}\,dt  
&= 
- \frac12
\sum_{j} \pounds_{\xi_j(x)}^T\left(\pounds_{\xi_j(x)}^Tp\right)  \, dt
\,,
\end{split}
\label{ItoEP-eqns}
\end{align}
where we have used $[dW_i(t),dW_j]=\delta_{ij}dt$ for Brownian motion and rearranged as shown in \eqref{Ito-comp} below,  in order to rewrite the Lie derivatives in terms of the stochastic It\^o vector field, $d\widehat{x}_t $, given in equation \eqref{vel-ydt}.

\begin{remark}\rm
The right hand sides in the It\^o stochastic equations \eqref{ItoEP-eqns} define a second-order  operator via double Lie derivatives with respect to the POD vector fields, ${\xi_j(x)}$, $j=1,2,\dots,K$. This seems like a type of the Laplace operator based on double Lie derivatives. To follow this intuition, we first introduce a \emph{Lie-dual} $\delta_{\xi_j}$ to the exterior derivative ${\MM d}$, defined by the differential form operations
\begin{align}
\delta_{\xi_j} q:
= 
{\xi_j}\contract ({\MM d}\  ({\xi_j}\contract q )) 
\quad\hbox{or, in other notation,}\quad
\delta_{\xi_j} q:
=\iota_{\xi_j}{\MM d} (\iota_{\xi_j} q)
\,, 
\label{Lie-dual-differential}
\end{align}
where ${\xi_j}\contract$ and $\iota_{\xi_j}$ denote two standard expressions for the operation of insertion of a vector field into an arbitrary differential $k$-form $q\in\Lambda^k$ and $\{\xi_j\}$ is taken as the given basis of POD vector fields associated with the spatial correlations of the cylindrical noise, enumerated in increasing order of eigenvalue by the subscript $j$.
In terms of the Lie-dual operation $\delta_{\xi_j}$ and the exterior differential operator $\MM d$,  we define the \emph{Lie-Laplacian} operator $\Delta_{Lie}$ by  
\begin{align}
\sum_{j} \pounds_{\xi_j}(\pounds_{\xi_j}q)
=
\sum_{j}  (\delta_{\xi_j}{\MM d} + {\MM d}\,\delta_{\xi_j} )q
=:
\Delta_{Lie} q
\,.
\label{LieLaplacian}
\end{align}
\end{remark}

\begin{proposition}\rm\label{d-Delta-comm-rel-prop}
The Lie Laplacian operator $\Delta_{Lie}$ commutes with the exterior differential operator $\MM d$. That is, 
\begin{align}
[\Delta_{Lie},\MM d] = \Delta_{Lie}\MM d - \MM d \Delta_{Lie} = 0\,.
\label{d-Delta-comm-rel-eqn}
\end{align}
\end{proposition}
\begin{proof}
This commutation relation follows immediately from the definition of the Lie Laplacian operator $\Delta_{Lie}$ in \eqref{LieLaplacian} and the property of the exterior differential that $\MM d^2q = 0 $ when acting on a differentiable $k$-form $q$. It also follows from the definition \eqref{LieLaplacian} because Lie derivatives commute with exterior derivatives. 
\end{proof}
\begin{remark}[Derivation of the Lie Laplacian operator in equation \eqref{ItoEP-eqns}]$\,$\rm

Let us seek the It\^o form of the stochastic part of the Stratonovich evolution in \eqref{SEP-eqns-thm-again}, ignoring drift. The stochastic part of the process in \eqref{SEP-eqns-thm-again} can be written as a linear differential operator 
\begin{align}
dq(x,t) = \pounds_{\xi (x)}q(x,t) \circ dW(t)
\,,
\label{SEP-eqns-thm-abbrev}
\end{align}
in which we have ignored the drift term and suppressed indices on the vector fields ${\xi_i(x)}$ for simplicity of notation. Upon pairing this equation with a \emph{time independent} test function $\phi(x)\in V^*$, we find the weak form of the Stratonovich advection equation \eqref{SEP-eqns-thm-abbrev},
\begin{align}
d \left\langle 
\phi(x)\,,\, q
\right\rangle_V 
= 
\left\langle 
\phi(x)\,,\, dq
\right\rangle_V
=
 \left\langle 
\phi(x)\,,\, \pounds_{\xi (x)}q \circ dW(t)
\right\rangle_V 
= 
 \left\langle 
\pounds_{\xi (x)}^T  \phi(x)\,,\, q
\right\rangle_V \circ dW(t)
\,,
\label{weak-strat-vection}
\end{align}
where we have taken advantage of the parametric $x$-dependence in $\phi(x)$ to pass the evolution operation $d$ through it in the first step. 
Here, $\langle \,\phi\,,\,q\,\rangle_V=\int <\phi(x),q(x,t)>dx$ denotes $L^2$ integral pairing of a quantity $q\in V$ with its tensor dual $\phi\in V^*$ in the domain of flow, the symbol $\pounds_{\xi (x)}^T $ is the $L^2$ adjoint of the Lie derivative $\pounds_{\xi (x)}$, and $<\cdot\,,\,\cdot>$ denotes pairing between elements of $V$ and their dual elements in $V^*$ at each point $x$ in the domain.
The corresponding weak form of the It\^o evolution is
\begin{align}
\begin{split}
d \left\langle 
\phi\,,\, q
\right\rangle_V 
&= 
 \left\langle 
\pounds_{\xi (x)}^T  \phi\,,\, q
\right\rangle_V  dW(t)
+
\frac12  \left[
 d \left\langle 
\pounds_{\xi (x)}^T  \phi\,,\, q
\right\rangle_V  
\,,\, dW(t)
\right]
\\
\hbox{By equation \eqref{weak-strat-vection} }&=
\left\langle 
\pounds_{\xi (x)}^T  \phi\,,\, q
\right\rangle_V  dW(t)
+
\frac12  \left[
 \left\langle 
\pounds_{\xi (x)}^T(\pounds_{\xi (x)}^T  \phi)\,,\, q
\right\rangle_V  
dW(t)\,,\, dW(t)
\right]
\\&=
 \left\langle 
  \phi\,,\, \pounds_{\xi (x)} q
\right\rangle_V  dW(t)
+
\frac12  \left[
 \left\langle 
\phi \,,\, \pounds_{\xi (x)} (\pounds_{\xi (x)} q)
\right\rangle_V  
dW(t)\,,\, dW(t)
\right]
\\&=
 \left\langle 
  \phi\,,\, \pounds_{\xi (x)} q
\right\rangle_V  dW(t)
+
\frac12  
 \left\langle 
\phi \,,\, \pounds_{\xi (x)} (\pounds_{\xi (x)} q)
\right\rangle_V  dt
\,,\end{split}
\label{Ito-comp}
\end{align}
Hence, because $\phi(x)\in V^*$ was chosen arbitrarily, the It\^o form of the Stratonovich stochastic advection equation for $q\in V$ in \eqref{SEP-eqns-thm-again} is
\begin{align}
dq + \pounds_{u}q\,dt  = \sum_i (\pounds_{\xi_i  (x)} q) dW_i (t) 
+ \frac12 \sum_{i,j}  \pounds_{\xi_j  (x)} (\pounds_{\xi_i  (x)} q)\,[dW(t)_i\,,\, dW_j(t)]
\,.
\label{Ito-form}
\end{align}
For Brownian motion, the last term in \eqref{Ito-form} simplifies via $[dW_i(t),dW_j]=\delta_{ij}dt$, and the middle equation in \eqref{ItoEP-eqns} emerges. Note that the last term in \eqref{Ito-form} (the quadratic It\^o term) cannot be written as a Lie derivative of the Stratonovich-advected quantity $q$. Instead, it is a \emph{double} Lie derivative, and this has the effect of masking the interpretation of the $q$-evolution in It\^o form as advection. 

\end{remark}

\section{Abstract Kelvin theorem}\label{KelThm-sec}

\subsection{Stratonovich circulation theorem}\label{KelThmStrat-subsec}

Next, we shall define the {\bfi circulation map\/}
$\mathcal K: {\mathcal C} \times V ^\ast
\rightarrow \mathfrak{X}(\mathcal{D})^{\ast\ast}$, where ${\mathcal C}$ is a space of \emph{material loops}, for which $c \in {\mathcal C}$ satisfies 
\begin{equation}
dc(t) = -\,\pounds_{dx_t}c(t)
\,.\label{c-eqn}
\end{equation}
Given a 1-form density $m \in \mathfrak{X} ^\ast $ we can
create a 1-form (no longer a density) by dividing it by the mass
density, $D$. We denote the result just by $m  / D$.
We let the circulation map $ {\mathcal K} $ then be defined by the integral of the 1-form $m/D$ around a loop moving with the Stratonovich flow,
\begin{equation}
 \left\langle {\mathcal K} (c(t) , q(t)) , m
\right\rangle  = \oint _{c(t)} \frac{ m }{ D } \, .
\label{circ.map}
\end{equation}
The expression in this definition is called the {\bfi circulation\/}
of the one-form $m/D$ around the loop $c(t)$ moving along a Stratonovich stochastic path $x_t(x_0)$  with \emph{Eulerian} velocity vector field $dx_t$ given as in equation \eqref{vel-vdt},
\begin{equation}
dx_t = u(x,t)\,dt - \sum_i \xi_i(x)\circ dW_i(t) \,.
\label{stochpathdiff}
\end{equation}
Consider the stochastic flow equation for the loop integral $\oint_{c(t)} m/D $. 
The differential of this loop integral is the total differential of the integrand, since the loop $c(t)$ itself is moving with the stochastic flow as in \eqref{c-eqn}. Consequently, we find
\begin{align}
\begin{split}
d \oint_{c(t)} m/D &= \oint_{c(t)} (d+ \pounds_{dx_t}) (m/D) 
= \oint_{c(t)} \frac{1}{D}\frac{\delta \ell}{\delta q}\diamond q \,dt
\,,\end{split}
\label{KN-loop}
\end{align}
where we have used equation \eqref{SEP-eqns-thm} in Theorem \ref{SEP-thm}  as well as the advection relation for density $D$,
\[
dD = -\,\pounds_{dx_t }D
=
-\,\pounds_{u(x,t)}D\,dt + \sum_i \pounds_{\xi_i(x)}D\circ dW_i(t)
\,.
\]
This calculation has proven the following theorem for the {\bfi Kelvin-Noether circulation map}
$I : {\mathcal C} \times V ^\ast  \times \mathfrak{X}
\rightarrow \mathbb{R}$ defined by
\begin{equation}
I(t) = I(c, q,u) := 
\left\langle{\mathcal K} (c, q), \frac{ \delta l}{\delta u}( u , q) \right\rangle.
\label{KelvinNoetherQty}
\end{equation}

\begin{theorem}[Abstract Kelvin-Noether theorem for Stratonovich
Euler--Poincar\'e SPDEs]\rm \label{KelvinNoetherSPDEthm}$\,$

For $c(t) \in {\mathcal C}$, let $u(t), q(t)$ satisfy the reduced Stratonovich
Euler--Poincar\'e SPDEs in \eqref{SEP-eqns-thm} in Theorem \ref{SEP-thm}. 
Choose the map  ${\mathcal K} : {\mathcal C} \times V ^\ast \rightarrow \mathfrak{X} ^{\ast \ast} $ given by integration around a loop $c(t)$ satisfying \eqref{c-eqn}.
Then the Kelvin-Noether circulation map in \eqref{KelvinNoetherQty} satisfies
\begin{equation}
dI(t) = \left\langle {\mathcal K}(t),
                     \frac{1}{D}\frac{\delta l}{\delta q} \diamond q \right\rangle\,dt\,.
\label{KNthm-SPDEs}
\end{equation}
\end{theorem}

\begin{remark}\rm
Upon using the definition \eqref{circ.map} for the Kelvin-Noether quantity in equation \eqref{KelvinNoetherQty}, one recovers the explicit formula \eqref{KN-loop} from the abstract Kelvin-Noether theorem in \eqref{KNthm-SPDEs}.

In particular, in the case of incompressible flow where $q=D=1$ (volume preservation) and $\delta l/\delta q=0$ for $q\ne D$, then the Kelvin-Noether circulation map in \eqref{KelvinNoetherQty} is \emph{conserved}, since upon using the definitions of the circulation in \eqref{circ.map} and the diamond operation in \eqref{diamond-def-} we find,
\begin{equation}
d\oint_{c(t)} \frac{\delta \ell (u)}{\delta u} 
= \oint_{c(t)} (d+ \pounds_{dx_t}) \frac{\delta \ell (u)}{\delta u}
= 0\,,
\label{KNthm-StochEuler}
\end{equation}
for any co-moving loop $c(t)$, i.e., any loop satisfying equation \eqref{c-eqn}, with stochastic vector field $dx_t$ defined in \eqref{stochpathdiff}. 

\end{remark}

\subsection{It\^o circulation theorem}\label{KelThmIto-subsec}

As we have seen in equation \eqref{KN-loop}, the Kelvin circulation theorem for the Stratonovich case is,
\begin{align}
d \oint_{c(t)} \frac{1}{D}\frac{\delta  \ell}{\delta  u} 
= \oint_{c(t)}  \frac{1}{D}\frac{\delta \ell}{\delta q}\diamond q \,dt
\,,
\label{Kel-Strat}
\end{align}
for loops $c(t)$ satisfying equation \eqref{c-eqn} with the stochastic path velocity vector field $dx_t$  defined in \eqref{stochpathdiff}.

For the It\^o case, according to equation \eqref{ItoEP-eqns} this circulation law becomes
\begin{align}
d \oint_{\widehat{c}(t)} \frac{1}{D}\frac{\delta  \ell}{\delta  u}  
= \oint_{\widehat{c}(t)}  \frac{1}{D}\frac{\delta \ell}{\delta q}\diamond q
+\oint_{\widehat{c}(t)}
\frac12 \sum_{i,j} \pounds_{\xi_j(x)}
\left(\pounds_{\xi_i(x)} \frac{1}{D}\frac{\delta  \ell}{\delta  u} \right)  \,\delta_{ij}\,dt
\label{Kel-Ito}
\end{align}
for loops $\widehat{c}(t)$ following paths generated by flows of the It\^o stochastic vector field $d\widehat{x}_t$ given in equation \eqref{vel-ydt}.

Hence, the It\^o equations have extra sources of circulation, even in the absence of advected quantities, $q$.
In particular, the Kelvin-Noether result in Theorem \ref{KelvinNoetherSPDEthm} shows that circulation is conserved for loops moving along the stochastic Stratonovich path when the advected quantities $q$ are absent. However, not unexpectedly, because the velocities of the loops $c(t)$ and $\widehat{c}(t)$ are \emph{different}, the Kelvin-Noether theorem shows that circulation is conserved \emph{only} for loops moving along the stochastic Stratonovich path, $c(t)$, and this conservation is \emph{masked} in the It\^o representation, because it does not hold for loops moving along the stochastic It\^o path, $\widehat{c}(t)$, and the quadratic covariation terms cannot be expressed as a single Lie derivative operation. Thus, as viewed along the It\^o path, $\widehat{c}(t)$, misalignment of the correlation eigenvectors creates or destroys circulation. 

\section{Examples}\label{Examples}

\subsection{Stratonovich stochastic fluid flows without advected quantities}\label{SEP-incompress-subsec}

\subsubsection{Stratonovich stochastic  Euler--Poincar\'e flows}

The Stratonovich form of the \emph{stochastic} Euler--Poincar\'e equations \eqref{SEP-eqns-thm} in Theorem \ref{SEP-thm} is, in the absence of advected quantities, 
\begin{align}
dv + \pounds_{d{x}_t}v = -\MM d p\,dt
\,,
\label{SEP-Eul-alpha-eqns}
\end{align}
in which bold $\MM d$ denotes the \emph{differential} in the standard exterior calculus sense, while italic $d$ still denotes the stochastic change, including both the drift and the stochastic process.  Here the variational derivative $v=\frac{\delta \ell}{\delta u}\in\Lambda^1$ is the momentum 1-form dual to the velocity vector field. The divergence-free velocity vector field $u=\MM u\cdot\nabla\in\mathfrak{X}(\mathbb{R}^3)$ in the covariant basis defined by the spatial gradient $(\nabla)$ has vector components also written in bold as $\MM u\in \mathbb{R}^3$ which satisfy ${\rm div}\,\MM u=0$. 
In fact, we shall assume that ${\rm div}\,\MM{\xi}_i=0$ with $\MM{\xi}_i\in \mathbb{R}^3$, as well, so we may write the divergence of the stochastic path in equation \eqref{stochpathdiff} in vector form as
\begin{align}
{\rm div} (d\MM{x}_t)= 0
\,,\quad\hbox{with}\quad
d\MM{x}_t 
= 
\MM u(\MM x,t)\,dt - \sum_i \MM{\xi}_i(\MM x)\circ dW_i(t)
\,.\label{vel-vectordt}
\end{align}
This means the stochastic flow of the vector field $d\MM{x}_t$ preserves volume elements. In this case, we may define the variational derivative 1-form density $\frac{\delta \ell}{\delta u}$ in Eulerian spatial coordinates as simply the 1-form
\[
v := \frac{\delta \ell}{\delta u} = \MM {v\cdot dx}\,.
\]
The Lie derivative in equation \eqref{SEP-Eul-alpha-eqns} is then written in vector form,  so that  equation \eqref{SEP-Eul-alpha-eqns} becomes
\[
(d\MM {v})\,\cdot\, \MM {dx} + \pounds_{dx_t}(\MM {v}\cdot \MM{dx} )
=
\Big(d \MM v + d\MM{x}_t\cdot \nabla \MM v + (\nabla d\MM{x}_t)^T\cdot \MM v
\Big)\,\cdot\, \MM {dx}
= -\,\nabla p \,\cdot\, \MM {dx} \,dt
\,.\]
In three-dimensional vector form, this equation is
\begin{align}
\begin{split}
d \MM v + d\MM{x}_t\cdot \nabla \MM v + (\nabla d\MM{x}_t)^T\cdot \MM v
&= -\,\nabla p \,dt
\\&=
d \MM v - d\MM{x}_t\times{\rm curl}\, \MM v + \nabla (d\MM{x}_t\cdot \MM v)
\,.
\end{split}
\label{v-vec-eqn}
\end{align}
Taking the curl of equation \eqref{v-vec-eqn} yields the stochastic equation for  the vorticity $\MM {\omega} = {\rm curl}\,\MM v$
\begin{align}
d\MM\omega - {\rm curl}\,\Big(d\MM{x}_t\times \MM\omega\Big) = 0
\,.\label{vort-vec-eqn}
\end{align}
We define the \emph{vorticity flux} as a 2-form $\omega$ with basis area element $\MM{dS}$ as
\begin{align}
\omega := \MM{d}v 
= \MM{d}(\MM v \cdot \MM {dx} )
= ({\rm curl}\,\MM v) \cdot \MM{dS}
= \MM\omega \cdot \MM{dS}
\,.\label{vort-2form-eqn}
\end{align}
Consequently, we may write the vorticity vector equation \eqref{vort-vec-eqn} in geometric form as
\begin{align}
\begin{split}
&d\omega + \pounds_{d{x}_t}\omega = 0
\,.
\end{split}
\label{SEPvorticity-eqn}
\end{align}
This equation also follows directly from the exterior derivative of the geometric form of the motion equation in \eqref{SEP-Eul-alpha-eqns} by invoking commutation of the Lie derivative and the exterior derivative. It means the flux of vorticity $\MM\omega \cdot \MM{dS}_t$ through a surface element following the stochastic path $x_t(x)$ is invariant. For an alternative formulation of the stochastic fluid vorticity equations in terms of a nonlinear version of the Feynman-Kac formula, see \cite{CrQi2013}.

\begin{remark}[ {\bf It\^o stochastic Euler--Poincar\'e equations}]
\label{StratvIto-conservlaws}\rm
The It\^o forms of the motion equation in \eqref{SEP-Eul-alpha-eqns} and vorticity equation \eqref{SEPvorticity-eqn} for incompressible motion in the absence of advected quantities are given by
\begin{align}
\begin{split}
&dv + \pounds_{d\widehat{x}_t}v = -\,\MM d p 
+ \frac12
\sum_{i,j} \pounds_{\xi_j(x)}(\pounds_{\xi_i(x)}\,v)   \,\delta_{ij}\,dt
\,,\\
&d\omega + \pounds_{d\widehat{x}_t}\omega = 
\frac12
\sum_{i,j} \left(\pounds_{\xi_j(x)}(\pounds_{\xi_i(x)}\,\omega)\right)  \,\delta_{ij}\,dt
\,,
\end{split}
\label{ItoMotionVorticity-eqns}
\end{align}
where the vorticity 2-form $\omega=\MM\omega \cdot \MM{dS}$ is given in \eqref{vort-2form-eqn}, and we have used commutation of exterior derivative $\MM d$ and Lie derivative $\pounds_{\xi_j(x)}$ twice.

\end{remark}

\subsubsection{Helicity preservation for Stratonovich stochastic Euler--Poincar\'e (SEP) flows}
\begin{definition} [Helicity]\rm
The helicity $\Lambda[{\rm curl}\,\MM{v}]$ of a divergence-free  vector field ${\rm curl}\,\MM{v}$ that is tangent to the boundary $\partial D$ of a simply connected domain $D\in\mathbb{R}^3$ is defined as
\begin{equation} 
\Lambda[{\rm curl}\,\MM{v}]
= 
\int_D (\MM{v} \cdot {\rm curl}\,\MM{v})\,{d}^3x
\,,
\label{helicity-def}
\end{equation} 
where $\MM{v}$ is a divergence-free vector-potential for the field ${\rm curl}\,\MM{v}$ and ${d}^3x$ is the spatial volume element. 
\end{definition}
\begin{remark}\rm
The helicity of a vector field ${\rm curl}\,\MM{v}$ measures the average linking of its field lines, or their relative winding. Refer to \cite{ArKh1998,MoTs1992} for excellent historical surveys.
The helicity is unchanged by adding a gradient  to the vector $\MM{v}$, and ${\rm div}\,\MM{v}=0$ is not a restriction for simply connected domains in $\mathbb{R}^3$, provided ${\rm curl}\,\MM{v}$ is tangent to the boundary $\partial D$.
\end{remark}

The principal feature of this concept for Stochastic Euler flows is embodied in the following theorem.

\begin{theorem}[Stratonovich stochastic Euler flows preserve helicity]\label{helicitycons-thm}$\,$\rm

When homogeneous or periodic boundary conditions are imposed, Euler's equations for an ideal incompressible fluid flow preserves the helicity, defined as the volume integral 
\begin{equation} 
\Lambda[{\rm curl}\,\MM{v}]
= 
\int_D \MM{v} \cdot {\rm curl}\,\MM{v}\,{d}^3x
=
\int_D v\wedge  {d}v
\,,
\label{helicity-thm}
\end{equation} 
where $v=\delta\ell/\delta u = \MM{v} \cdot \MM {dx}$ is the circulation 1-form, 
$\MM {d}v = {\rm curl}\,\MM{v} \cdot \MM {dS}$ is the vorticity flux (a 2-form),
${\rm curl}\,\MM{v}=\MM {\omega}$ is the vorticity vector and ${d}^3x$ is the spatial volume element.
\end{theorem}

\begin{proof}
Rewrite the geometric form of the Stochastic Euler equations \eqref{SEP-Eul-alpha-eqns} for rotating incompressible flow with unit mass density in terms of the circulation one-form $v:=\MM{v}\cdot \MM {dx}$ as 
 \begin{equation}
dv + \pounds_{dx_t}v = -\MM{d} p \,dt
\,.
\label{SEule-qns-1form}
\end{equation}
and $\pounds_{dx_t}\,{d}^3x = {\rm div} (d\MM{x}_t)\,{d}^3x =0$. 
Then the {\bfi helicity density}, defined as 
 \begin{equation}
 v\wedge \MM{d}v=(\MM{v}\cdot{\rm curl}\,\MM{v})\,{d}^3x
 =\lambda\,{d}^3x\,,
 \quad\hbox{with}\quad
 \lambda= \MM{v}\cdot{\rm curl}\,\MM{v}
 \,,
\label{helicity dens-def}
\end{equation}
obeys the dynamics it inherits from the Stochastic Euler equations,
 \begin{equation}
\big(d + \pounds_{dx_t}\big)(v\wedge \MM{d}v)
= - (\MM dp \wedge \MM{d}v)\,dt - (v\wedge \MM d^2p)\,dt = - (\MM d (p\, \MM{d}v))\,dt
\,,
\label{helicity dens-eqn-Lie}
\end{equation}
after using $\MM d^2p=0$ and $\MM d^2v=0$.
In vector form, this result may be expressed as a conservation law,
 \begin{equation}
\big(d \lambda + {\rm div}\, \lambda d\MM{x}_t\big)\,{d}^3x
=
-\,{\rm div}(p\,{\rm curl}\,\MM{v})\,{d}^3x\,dt
\,.
\label{helicity dens-eqn-vec}
\end{equation}
Consequently, the time derivative of the integrated helicity in a domain $D$ obeys
 \begin{eqnarray}
d\Lambda[{\rm curl}\,\MM{v}]
 &=& 
 \int_D (d\lambda)\,{d}^3x
 = 
- \int_D {\rm div}( \lambda d\MM{x}_t + p\,{\rm curl}\,\MM{v}\,dt)\,{d}^3x  
\nonumber\\
&=& 
- \oint_{\partial D}( \lambda d\MM{x}_t + p\,{\rm curl}\,\MM{v}\,dt)\cdot\MM{\widehat{n}} \,{\MM{d}S} 
\,,
\label{Euler-helicity-thm}
\end{eqnarray}
which vanishes when homogeneous or periodic boundary conditions are imposed on $\partial D$.
\end{proof}

\begin{corollary}[The It\^o representation of stochastic Euler flows masks the Stratonovich preservation of helicity]\label{helicitynoncons-cor}
\end{corollary}

\begin{proof}
Equation \eqref{ItoMotionVorticity-eqns} for It\^o Euler--Poincar\'e incompressible flow yields, for $\omega=\MM{d}v$ as before,
\begin{align}
\begin{split}
(\partial_t + \pounds_{d\widehat{x}_t/dt})(v\wedge\omega)  & =
\left(-\MM dp + \frac12 \Delta_{Lie} v\right)\wedge\omega
+ v\wedge \frac12 \Delta_{Lie}\omega
\\ & =
-\,\MM d \left(p \omega + v\wedge\frac12 \Delta_{Lie} v\right)
+ v\wedge \Delta_{Lie}\omega
\,.\end{split}
\label{ItoHelicity-eqn}
\end{align}
In vector form, the last equation in \eqref{ItoHelicity-eqn} may be integrated over space and written as
\begin{align}
\frac{d}{dt} \int_D ({\MM v} \cdot \MM {\omega})\,{d}^3x = 
-\int_D {\rm div}\left(
({\MM v} \cdot \MM {\omega})\, \frac{d\MM{\widehat{x}}_t}{dt} + p\MM {\omega} 
+  \frac12 \MM v\times \Delta_{Lie}\MM {\omega}
\right){d}^3x
+ 
\int_D
(\MM v \cdot \Delta_{Lie}\MM {\omega} )
\,{d}^3x
\,.
\label{ItoHelicity-vectoreqn}
\end{align}
\end{proof}
\begin{remark}\label{caution}\rm
Hence, even for homogeneous or periodic boundary conditions, in which the integral of the divergence would vanish, there remains a non-vanishing term on the right hand side of equation \eqref{ItoHelicity-vectoreqn} for the It\^o evolution of the helicity, due to the additional quadratic drift term arising in the It\^o calculus. Thus, the It\^o stochastic dynamics appears to predict reconnection of vorticity field lines, although their linkages are actually \emph{preserved} in the Stratonovich representation. Hence, as usual, one must take caution in drawing conclusions about stochastic fluid dynamics, because some of its features may be representation-dependent.  
\end{remark}

\subsection{The effects of advected quantities}\label{Diamond-operator}
In this section, we compute the explicit formulae needed in applications of the system of stochastic equations in \eqref{SEP-eqns-thm} for the vector space $V$ of 3D quantities $q\in V$ consisting of elements with
the following coordinate functions in 3D Euclidean vector notation, 
\begin{equation}
q\in\{b,\mathbf{A}\cdot \MM {dx},\mathbf{B}\cdot \MM {dS},D\,{d}^3x\} =: V\,.
\label{Eul-ad-qts}
\end{equation}
Dual quantities under the $L^2$ pairing are $(b,D\,{d}^3x)$ and $(\mathbf{A}\cdot \MM {dx},\mathbf{B}\cdot \MM {dS})$.
The vector space $V$ contains the geometric quantities that typically occur in ideal continuum dynamics. 
These are: scalar functions $(b)$, 1-forms $(\mathbf{A}\cdot \MM {dx})$,
2-forms $(\mathbf{B}\cdot \MM {dS})$, and densities $(D\,{d}^3x)$ in three dimensions.
In addition, with applications to magnetohydrodynamics (MHD) in mind, we also choose
$\mathbf{B}={\rm curl}\,\mathbf{A}$ and
$\MM d(\mathbf{A}\cdot \MM {dx})=\mathbf{B}\cdot \MM {dS}$.
In Euclidean
coordinates on $\mathbb{R}^3$, this is $\MM{d}(A_k{d}x^k)=A_{k,j}{d}x^j\wedge {d}x^k
=\frac{1}{2} \epsilon_{ijk}B^i{d}x^j\wedge {d}x^k$, where $\epsilon_{ijk}$
is the completely antisymmetric tensor density on $\mathbb{R}^3$ with
$\epsilon_{123}=+1$. The two-form
$\mathbf{B}{\cdot}\MM {dS}=\MM d(\mathbf{A}\cdot \MM {dx})$
is the physically interesting
special case of $B_{kj}{d}x^j{\wedge}{d}x^k$ for MHD, in which
$B_{kj}=A_{k,j}$, so that $\nabla\cdot\mathbf{B}=0$. 
\begin{definition}[Deterministic advection relations for $q\in V$ in \eqref{Eul-ad-qts}]\label{diamond-def-q}\rm
The deterministic advection relations $\partial_t q = -\, \pounds_u q$ for the quantities $q\in V$ in \eqref{Eul-ad-qts} 
are given explicitly by the Lie-derivative action of smooth vector fields $u\in \mathfrak{X}(\mathbb{R}^3)$ 
on the vector space of variables $q\in V$. These deterministic advection relations are given by
\begin{align}
\begin{split}
\partial_t  b
&= -\pounds_u\, b = -\,\MM{u}\cdot\nabla\,b\,,
\\
\partial_t  \mathbf{A}\cdot \MM {dx}
&= -\pounds_{u}\,(\mathbf{A}\cdot \MM {dx})
=-\left((\MM{u}\cdot\nabla)\mathbf{A}+A_j\nabla u^j\right)
\cdot \MM {dx}
\\
&= \left(\MM{u}\times{\rm curl}\,\mathbf{A}
-\nabla(\MM{u}\cdot\mathbf{A})\right)\cdot \MM {dx}\,,
\\
\partial_t \mathbf{B}\cdot \MM {dS}
&=-\pounds_{u}\,(\mathbf{B}\cdot
\MM {dS})
= \left({\rm curl}\,(\MM{u}\times\mathbf{B})\right)\cdot \MM {dS}
\ =\ d (\partial_t  \mathbf{A}\cdot \MM {dx})
\,,
\\
\partial_t  D\ {d}^3x&=-\pounds_{u}\,(D\,{d}^3x)
= -\nabla\cdot({D\bf u})\ {d}^3x\,.
\end{split}
\label{Lie-derivs}
\end{align}
\end{definition}

\paragraph{The diamond operation.}
The diamond operation $(\diamond): T^*V\to \mathfrak{X}^*$ is defined for $(q,p)\in T^*V$ and $u\in \mathfrak{X}(\mathbb{R}^3)$ by equation \eqref{diamond-def-} as
\begin{equation}
\left\langle  p\diamond q\,,\,u \,\right\rangle_\mathfrak{X} = \left\langle  p\,,\,- \,\pounds_u q\,\right\rangle_V 
,\label{diamond-def-eqn}
\end{equation}
for the $L^2$ pairings $\langle  \,\cdot\,,\,\cdot\,\rangle_V: T^*V\times TV\to \mathbb{R}$ and
$\langle  \,\cdot\,,\,\cdot\,\rangle_\mathfrak{X}: \mathfrak{X}^*\times \mathfrak{X}\to \mathbb{R}$
with $p\diamond q\in \mathfrak{X}^*$. Under the $L^2$ pairing, we assume that boundary terms arising from
integrations by parts may be dropped, by invoking natural boundary conditions.

In particular, for the set of advected quantities $q\in V$ in \eqref{Eul-ad-qts} above, we
find the following Euclidean components of the sum of terms 
$\frac{\delta \ell}{\delta q}\diamond q$ in the motion equation \eqref{SEP-eqns-thm} in Theorem \ref{SEP-thm},
for stochastic EP equations,
\begin{align} 
\begin{split}
\left(    \frac{\delta \ell}{\delta q}\diamond q \right)_k
&=
-\ 
\frac{ \delta \ell}{ \delta b}(\nabla b)_k
\ +\ D\left(\nabla\frac{ \delta \ell}{ \delta D}\right)_k
+ 
\left(
-\ \frac{ \delta \ell}{ \delta \mathbf{A}}\times{\rm curl}\, \mathbf{A}
   +\mathbf{A}\ {\rm div}\,\frac{ \delta \ell}{ \delta \mathbf{A}}\
   +\mathbf{B}\times{\rm curl}\, \frac{ \delta \ell}{ \delta \mathbf{B}}
\right)_k
.
\end{split}
\label{diamond-q-eq}
\end{align}
%
With these definitions, one may write explicit formulae for the fluid examples needed in applications of the system of stochastic equations in \eqref{SEP-eqns-thm} for deterministic advected quantities $q$ in the vector space $V$ of three-dimensional quantities in \eqref{Eul-ad-qts}. These applications include, for example, geophysical fluid dynamics (GFD) and magnetohydrodynamics (MHD). The applications of the present theory to MHD will be pursued elsewhere. In this paper, we restrict ourselves to examples from GFD.  

\subsection{Stochastic geophysical fluid dynamics (SGFD)}\label{GFD-subsec}

\subsubsection{Euler-Boussinesq approximation}

\paragraph{Stochastic Euler-Boussinesq equations of a rotating stratified  incompressible fluid.}
In SGFD, the stochastic Euler--Poincar\'e (SEP) equations in \eqref{SEP-eqns-thm} are found in the Euler-Boussinesq approximation, for example, by choosing the Lagrangian $\ell({u},b,D)$ depending on the set $q\in\{b,D\}$ and given by \cite{HoMaRa2002}
\begin{align}
\ell(u,b,D) = \int \frac{D}{2} |u|^2 + Du\cdot R - gbDz - p(D-1)\,{d}^3x
\,,
\label{GFD-Lag}
\end{align}
where $u$ is fluid velocity, $D$ is the volume element, $2\Omega={\rm curl}R(x)$ is  the Coriolis vector, while $R(x)$, a given function of ${x}$ is its vector potential, $g$ is the constant gravitational acceleration, $b$ is buoyancy, and the pressure $p$ is a Lagrange multiplier which enforces $D=1$, so that ${\rm div}u=0$. The GFD Lagrangian \eqref{GFD-Lag} possesses the following variations at fixed ${x}$ and $t$,
\begin{align}
\begin{split}
\frac{m}{D}
&=
\frac{1}{D}\frac{{\delta} l}{{\delta} {u}}
= {u}+ {R}(x)
\,,
\quad
\frac{{\delta} l}{{\delta} b}
= - Dgz
\,,
\\
\frac{{\delta} l}{{\delta} D}
&=  \frac{1}{2}|{u}|^2 + {u}\cdot{R} - gzb - p
\,,\quad
\frac{{\delta} l}{{\delta} p} 
= -\,(D-1)\,.
\end{split}
\label{vds-1}
\end{align}
Hence, from the stochastic Euler--Poincar\'e (SEP) equations in \eqref{SEP-eqns-thm}, we find the motion equation for an Euler-Boussinesq fluid in three dimensions,
\begin{align}
du + \pounds_{dx_t}(u + R(x))
=
-gb\,\nabla z \,dt + \nabla\left(-p + \frac{1}{2}|{u}|^2+ {u}\cdot{R}\right)dt 
\,,\quad\hbox{and}\quad
db + \pounds_{dx_t}b  = 0
\,,
\label{SEP-eqns-thm-again2}
\end{align}
with $d\MM{x}_t$ given as before in \eqref{vel-vdt} and \eqref{stochpathdiff}.%
\footnote{In this section, we distinguish between a vector field $dx_t\in \mathfrak{X}(\mathbb{R}^3)$ and its vector representation in a Cartesian basis, $d\MM x_t\in\mathbb{R}^3$. For example, with this notation one may write $dx_t=d\MM x_t\cdot\nabla 
$.}
Consequently, we find the stochastic advection law,
\begin{align}
(d+ \pounds_{dx_t})Q  = dQ + d\MM{x}_t\cdot\nabla Q = 0
\,,
\label{SEP-PV-eqn}
\end{align}
for the potential vorticity, $Q$, defined by its traditional formula,
\begin{align}
Q := {\rm curl}(u+R(x))\cdot \nabla b \,,
\label{SEP-PV-def}
\end{align}
where the total vorticity $\omega={\rm curl}(u+R(x))$ satisfies 
\begin{align}
d\omega = {\rm curl}(dx_t\times \omega) -g\nabla b\times\nabla z \,dt
\,,\quad\hbox{with}\quad
db + dx_t\cdot\nabla b = 0
\label{SEP-omega-eqn}
\,.\end{align}
These results are implied by the Stratonovich circulation theorem in Section \ref{KelThmStrat-subsec}. The stochastic conservation law for the potential vorticity $Q$ in equation \eqref{SEP-PV-eqn}, means that $Q$ is preserved along \emph{each} Stratonovich stochastic path.  

\paragraph{It\^o form of the potential vorticity equation.} The interpretation of the It\^o form of the potential vorticity equation \eqref{SEP-PV-eqn} may be obtained by expanding it out using the It\^o equations in \eqref{ItoEP-eqns}. Indeed, the It\^o form of the Stratonovich equation \eqref{SEP-PV-eqn} for potential vorticity $Q$ is the precisely the same as the double Lie-derivative formula for the advected quantity $q$ in equation \eqref{ItoEP-eqns}; namely, 
\begin{equation}
dQ + \pounds_{d\widehat{x}_t}Q\,dt 
= 
\frac12
\sum_{i,j} \pounds_{\xi_j(x)}(\pounds_{\xi_i(x)}Q)  \,\delta_{ij}\,dt
\,,
\end{equation}
or, in vector calculus form, 
\begin{equation}
dQ + {d\MM{\widehat{x}}_t}\cdot\nabla Q
= 
\frac12
\sum_{i,j} {\MM \xi}_j(\MM x)\cdot\nabla ({\MM \xi}_i(\MM x)\cdot\nabla\, Q)  \,\delta_{ij}\,dt
\,.
\label{SEP-PV-Itoeqn}
\end{equation}

Upon comparing equations \eqref{SEP-PV-eqn} and \eqref{SEP-PV-Itoeqn}, one concludes that the stochastic Euler-Boussinesq equations \eqref{SEP-eqns-thm-again2} in three dimensions preserve the potential vorticity $Q$ defined in equation \eqref{SEP-PV-def} along the Stratonovich stochastic path ${x}_t(x)$; but $Q$ preservation is \emph{masked} in the It\^o representation, because $Q$ is not preserved along the It\^o stochastic path ${\hat{x}}_t(x)$, which of course is a \emph{different} path.
The operator $\Delta_{Lie}$ in the last term in the It\^o equation \eqref{SEP-PV-Itoeqn} reduces to the metric Laplacian $\Delta=\nabla^2$ in the case that the vectors ${\MM \xi}_j$ with $j=1,2,3,$ are linearly independent unit vectors in three dimensions. 

\subsubsection{Quasigeostrophic (QG) approximation}\label{SQG-sec}

\paragraph{Deterministic quasigeostrophic (QG) equations.} 
The quasigeostrophic (QG) approximation is a fundamental model which is often used  for the analysis of meso- and large-scale motion in geophysical and astrophysical fluid dynamics \cite{Pedlosky87}. Physically, the QG approximation applies when the motion is nearly in geostrophic balance, i.e., when pressure gradients nearly balance the
Coriolis force. In the simplest case of a barotropic fluid
in a domain ${\mathcal{D}}$ on the plane ${\mathbb R}^2$ with coordinates
$(x_1,x_2)$, geostrophic balance determines the geostrophic fluid velocity $\MM{u}$ as 
$
\MM{u} := \mathbf{\widehat{z}}\times \nabla \psi
\,,
$
where $\psi$ is the stream-function, the flow is incompressible ($\nabla\cdot \MM{u}=0$) and $\mathbf{\widehat{z}}$ is the unit vector normal to the plane.

QG dynamics in the
$\beta $-plane approximation is expressed by the following evolution
equation for the stream-function $\psi$ of the geostrophic fluid velocity $\MM{u}$,
\begin{equation}
{\frac{\partial ({\Delta }\psi -{\mathcal{F}}\psi )}{\partial t}}
+[\psi ,{\Delta }\psi ]_{Jac}
+\beta {\frac{\partial \psi}{\partial x_1}} =0,  \label{qgb}
\end{equation}
where $\partial /\partial t$ is the partial time derivative, ${\Delta }$ is
the planar Laplacian, ${\mathcal{F}}$ denotes rotational Froude number, 
the square bracket $[\,\cdot\,,\,\cdot\,]_{Jac}$ denotes
\begin{equation}
[\,a\,,\,b\,]_{Jac} :=  {\frac{\partial (a,b)}{\partial (x_1,x_2)}}
= 
\mathbf{\widehat{z}}\cdot \nabla a \times \nabla b
\,,
\label{Jacobi-brkt}
\end{equation}
which is the Jacobi
bracket (Jacobian) for functions $a$ and $b$ on ${\mathbb R}^2$, and $\beta $
is the $x_2$-gradient of the Coriolis parameter, $f$, taken as $f=f_0+\beta x_2$ in
the $\beta$-plane approximation, with constants $\beta $ and $f_0$. Neglecting
$\beta $ gives the $f$-plane approximation. 

The QG equation (\ref{qgb}) may be
derived from the basic equations of rotating shallow water flow by rescaling to define
nondimensional variables and making an asymptotic expansion in terms of the Froude 
number given by the square of the ratio of the characteristic scale of the motion
to the  deformation radius, see, for example, Pedlosky \cite{Pedlosky87}, Allen and Holm \cite{AlHo1996}, and Zeitlin and Pasmenter \cite{ZePa1994}) for details of the QG derivation. 

Equation (\ref{qgb}) may be written alternatively  as an advection law for the \emph{potential vorticity}, $Q$,
\begin{equation}
{\frac{\partial Q}{\partial t}}
=
-\,\MM{u}\cdot \nabla Q
=
-\,\mathbf{\widehat{z}}\times \nabla \psi \cdot \nabla Q
\,,\quad \hbox{where}\quad
Q := {\Delta }\psi -{\mathcal{F}}\psi +f\,.  
\label{qgb'}
\end{equation}
This form of QG dynamics emphasises its basic property; namely, potential vorticity
$Q$ is conserved on geostrophic fluid parcels.

Much is known about the mathematical structure of the deterministic QG equations
(\ref{qgb}) and (%
\ref{qgb'}). In particular, they are Hamiltonian with a Lie-Poisson bracket
given in Weinstein \cite{AW83} as
\begin{equation}
\{F,H\} = -\int Q \Big[{\frac{{\delta} F}{{\delta} \mu}}, 
{\frac{{\delta} H}{{\delta} \mu}}\Big]_{Jac}\,\MM{d}x_1\wedge\MM{d}x_2
\,,  \label{qg-lpb}
\end{equation}
where $\mu:= Q-f$ and the square bracket $[\,\cdot\,,\,\cdot\,]_{Jac}$ is the Jacobi bracket in \eqref{Jacobi-brkt}. In terms of the variable $\mu$, the Hamiltonian for
QG is expressed as
\begin{align}
H 
= {\frac{1 }{2}}\int_{\mathcal{D}} \Big (|{\mbox{\boldmath{$\nabla$}}}
\psi|^2+{\mathcal{F}}\psi^2 \Big ) \,\MM{d}x_1\wedge\MM{d}x_2\ 
= {\frac{1 }{2}}\int_{\mathcal{D}} \mu\ ({\mathcal{F}}-{\Delta})^{-1} \mu \,\MM{d}x_1\wedge\MM{d}x_2
+ {\frac{1 }{2}}\sum_i\oint_{{\gamma}_i}\psi\,\MM{u}\cdot \MM {dx}\,,
\label{qg-ham}
\end{align}
where ${\gamma}_i$ is the $i$-th connected component of the boundary $\partial\mathcal{D}$. In what follows, we will discuss cases where the domain ${\mathcal{D}}$ is either a  torus (periodic boundary conditions) or the whole plane ${\mathbb R}^2$ with decaying boundary conditions and, thus, the boundary
terms may be ignored. Hence, we may write the Hamiltonian in terms of the $L^2$ pairing $\langle\,\cdot\,,\,\cdot\,\rangle_{L^2}$ as
\begin{align}
H (\mu) =\frac12 \left\langle \mu\,,\,({\mathcal{F}}-{\Delta})^{-1} \mu\right\rangle_{L^2}
\quad\hbox{with}\quad
\frac{\delta H}{\delta \mu} = ({\mathcal{F}}-{\Delta})^{-1} \mu = \psi
\,.
\label{qg-ham-br}
\end{align}
Consequently, the Lie-Poisson bracket (\ref{qg-lpb}) gives, after integration
by parts, the dynamical equation for $\mu$,
\begin{equation}
\frac{\partial \mu}{ \partial t} = \{ \mu ,H \}
= -\,[\psi, Q]_{Jac} = -\, {\bf u} \cdot
{\mbox{\boldmath{$\nabla$}}} Q\,,
\label{qgb2}
\end{equation}
in agreement with the QG potential vorticity equation (\ref{qgb'}).
Casimirs of the Lie-Poisson bracket (\ref{qg-lpb}) are given by
\begin{equation}
C_{\Phi}=\int \Phi(Q)\,\MM{d}x_1\wedge\MM{d}x_2
\,,
\label{QG-Casimirs}
\end{equation}
for an arbitrary function $\Phi$ and they satisfy
$\{C_{\Phi},H\}=0$ for all Hamiltonians
$H(\mu)$. Level surfaces of the Casimirs $C_{\Phi}$ define coadjoint orbits of
the group of symplectic diffeomorphisms of the domain of the flow \cite{AW83}.

\paragraph{Stochastic quasigeostrophic (SQG) equations.} 
Stochastic QG may be derived by applying the Lie-Poisson structure in \eqref{qg-lpb} to a  stochastic Hamiltonian. Thus, to introduce stochastic forcing into the QG equations we propose to augment the QG Hamiltonian in \eqref{qg-ham-br} with a Stratonovich stochastic term, as
\begin{equation}
h\,dt= \left\langle \mu\,,\,\frac12 ({\mathcal{F}}-{\Delta})^{-1} \mu\,dt 
- \sum_i\xi_i(\MM{x}) \circ dW_i
\right\rangle_{L^2}
\,,
\label{qg-augHam}
\end{equation}
in which the $\{\xi_i(\MM{x})\}$ are prescribed spatial functions. 
The variational derivative of the augmented QG Hamiltonian in \eqref{qg-augHam} yields the Stratonovich \emph{stochastic stream function}
\begin{equation}
\frac{\delta h}{\delta \mu}\,dt 
= \psi(\MM{x},t)\,dt - \sum_i\xi_i(\MM{x}) \circ dW_i
=: \Psi\,dt 
\,.
\label{qg-stochSF}
\end{equation}
Consequently, the Lie-Poisson bracket (\ref{qg-lpb}) gives, after integration
by parts, the dynamical equation for $\mu$,
\begin{equation}
d\mu= \{ \mu ,h\,dt \}
= -\,[\Psi \,dt, Q]_{Jac} 
= -\,[\psi \,,\, Q]_{Jac} \,dt + \sum_i [ \xi_i(\MM{x}), Q]_{Jac}  \circ dW_i 
= -\, \MM {dx}_t \cdot \nabla Q\,,
\label{qgb3}
\end{equation}
in which, cf. equation \eqref{stochpathdiff},
\begin{equation}
d\MM{x}_t := \mathbf{\widehat{z}}\times \nabla \Psi\,dt
= \mathbf{\widehat{z}}\times \nabla \left(
\psi\,dt - \sum_i\xi_i(\MM{x}) \circ dW_i
\right)
= \MM{u}\,dt - \sum_i
\Big(\mathbf{\widehat{z}}\times \nabla \xi_i(\MM{x}) \Big) \circ dW_i
\,.
\label{qg-dx}
\end{equation}
This result is in agreement with the advection law in the stochastic equation \eqref{SEP-PV-eqn} for potential vorticity in the Euler-Boussinesq approximation, provided we identify $\mathbf{\widehat{z}}\times \nabla \xi_i(\MM{x})$ for these stochastic QG equations with the $\mathbb{R}^3$ vector functions $\boldsymbol{\xi}_i(\MM{x})$ for the stochastic Euler-Bouusinesq equations in the previous example. 

Of course, the Casimirs of the Lie-Poisson bracket (\ref{qg-lpb}) given by $C_{\Phi}=\int \Phi(Q)\,\MM{d}x_1\wedge\MM{d}x_2$ for an arbitrary function $\Phi$ are still conserved along the Stratonovich stochastic path, since this property holds for the Lie-Poisson bracket, independently of the choice of a stochastic Hamiltonian. 

\paragraph{It\^o form.} The It\^o form of equation \eqref{qgb3} is,
\begin{equation}
d\mu
= -\,\big[\psi\,dt -  \sum_i \xi_i(\MM{x})\, dW_i \,,\, Q \big]_{Jac} 
+ \frac12 \sum_{j}\big[ \xi_j(\MM{x}),[ \xi_j(\MM{x}), Q]_{Jac}\,\big]_{Jac} \,dt\,,
\label{qgb-Ito}
\end{equation}
where we have assumed that the It\^o stochastic processes $dW_i(t)$ and $dW_j(t)$ are Brownian processes, whose quadratic covariations satisfy $[dW_i(t),dW_j(t)]=\delta_{ij}\,dt$ \cite{Sh1996}. 
This SQG version of the It\^o dynamics for potential vorticity $Q$ is reminiscent of potential vorticity dynamics for dissipative QG with viscosity.

\section{Conclusions}\label{conclus-sec}



The stochastically constrained variational principle $\delta S = 0$ for the action $S$ given in \eqref{SVP1} for introducing Stratonovich stochasticity into Euler-Poincar\'e equations for continuum dynamics has yielded stochastic incompressible flows that were found to preserve three important properties of ideal incompressible Euler flows which arise from its invariance under relabelling of Lagrangian coordinates. These three properties are: (i) the Kelvin circulation theorem in equation \eqref{KNthm-StochEuler}; (ii) invariance of the flux of vorticity in \eqref{SEPvorticity-eqn} through any surface element following the Stratonovich stochastic path $x_t$ in \eqref{stochpathdiff};  and (iii) preservation of the linkage number for the vorticity field lines (helicity) in Theorem \ref{helicitycons-thm}. The Kelvin circulation theorem is preserved because the material fluid loops follow the Stratonovich stochastic paths, as do the field lines of the vorticity, so all three results follow from the same interpretation. Likewise, the stochastic conservation law for the potential vorticity $Q$, found in equation \eqref{SEP-PV-eqn} for the Euler-Boussinesq equations with Stratonovich noise, meant that $Q$ was preserved along \emph{each} Stratonovich stochastic path. The same type of potential vorticity preservation along Stratonovich stochastic paths was found again for the case of stochastic quasigeostrophy (SQG) in Section \ref{SQG-sec}. 

In contrast, the It\^o representation of the stochastic equations involved a different stochastic vector field $d\widehat{x}_t$ in \eqref{vel-ydt} whose drift velocity contained the additional quadratic term that arises in It\^o calculus. This quadratic It\^o drift velocity introduced terms that were not present in the Stratonovich stochastic equations and could not be expressed as single Lie derivatives. As a result, the It\^o representation masked the preservation of ideal incompressible Euler flow properties and conservation of potential vorticity which was found to hold in the Stratonovich case. 

The It\^o drift term turned out to contain an interesting Laplacian-like operator, $\Delta_{Lie}$, defined in equation \eqref{LieLaplacian} by a sum over double Lie derivatives with respect to the POD vector fields, ${\xi_j(x)}$, $j=1,2,\dots,K,$ as
$\Delta_{Lie}:= \sum_{j} \pounds_{\xi_j(x)}(\pounds_{\xi_j(x)}\,\cdot\,)$,
which we called the \emph{Lie Laplacian} operator, since $\pounds_{\xi_j(x)}$ denotes the Lie derivative with respect to the vector field $\xi_j(x)$. This term is the generalisation for advected quantities $q\in V$, in an arbitrary vector space $V$, of the quadratic covariation drift term found already for scalar densities by Stratonovich \cite{St1966}. 
The Lie Laplacian operator appearing, for example, in the It\^o representation of the stochastic fluid equations \eqref{ItoEP-eqns} is not a standard Laplacian operator, although it reduces to the metric Laplacian when the independent vector fields $\xi_j(x)$ for stochastic spatial correlations are constant and $j=1,2,3$. It would be interesting to know whether (because of its relation to the metric Laplacian)  the Lie Laplacian operator in the It\^o representation could have a regularising effect on stochastic fluid equations which might otherwise be ill-posed, as suggested in \cite{Fl2011,FlMaNe2014}. 

\paragraph{Data accessibility} No data 
\paragraph{Competing interests} No competing interests
\paragraph{AuthorsÕ contributions} Sole author, only one contribution
\paragraph{Acknowledgements} 
I am very grateful for the encouragement of the many people who took the time to discuss these matters with me, or comment on drafts, especially my friends and colleagues N. Bou-Rabee, A. Castro, C. J. Cotter, D. Crisan, T. D. Drivas, B. K. Driver, F. Flandoli, F. Gay-Balmaz, T. G. Kurtz, J. P. Ortega, G. Pavliotis and T. Tyranowsky. However, as usual, any mistakes belong to the author. \paragraph{Funding statement} This work was partially supported by the European Research Council Advanced Grant 267382 FCCA.



\begin{thebibliography}{}

\bibitem{AlHo1996} 
J. S. Allen and D.~D. Holm [1996] Extended-geostrophic Hamiltonian
models for rotating shallow water motion,
\textit{Physica D} 98: 229--248.

\bibitem{ArChCr2012} 
M. Arnaudon, X. Chen and A. B. Cruzeiro [2014]
Stochastic Euler-Poincar\'e reduction,
\textit{J. Math. Physics} 55: 081507.

\bibitem{ArKh1998} 
Arnold, V. I. and Khesin, B. A. [1992] Topological methods in hydrodynamics, {\it Annu. Rev. Fluid Mech.} {\bf 24}, 145--166.

\bibitem{BeHoLu1993} 
G. Berkooz, P. Holmes, and J. L. Lumley [1993]
The Proper Orthogonal Decomposition in the analysis of turbulent flows. 
\textit{Annu. Rev. Fluid Mech.} 25: 539-575.

\bibitem{Bi1981}  
J. M. Bismut [1981] \textit{M\'ecanique al\'eatoire}, Berlin: Springer.

\bibitem{BR-O2009} 
N. Bou-Rabee and H. Owhadi [2009]
Stochastic variational integrators,
\textit{IMA Journal of Numerical Analysis} 29: 421--443.

\bibitem{ChCr2013}
X. Chen and A. B. Cruzeiro [2013]
Stochastic geodesics and forward-backward stochastic differential equations on Lie groups, 
Discrete and Continuous Dyn. Syst., Suppl., pp. 115--121.

\bibitem{CiCr1999}
F. Cipriano and A. B. Cruzeiro [1999]
Flows associated to tangent processes on the Wiener space,
J. Funct. Anal. 166: 310--331.

\bibitem{CiCr2007} 
F. Cipriano and A.B. Cruzeiro [2007] Navier--Stokes equation and diffusions on the group of homeomorphisms of the torus, Comm. Math. Phys. 275: 255--269.

\bibitem{CoIy2008} 
P. Constantin and G. Iyer [2008]
A stochastic Lagrangian representation
of the three-dimensional incompressible
Navier-Stokes equations,
Comm. Pure and Appl. Math., 61: 0330--0345.

\bibitem{CrQi2013} 
A.B. Cruzeiro and Zh. Qian [2013] Backward stochastic differential equations associated with the
vorticity equations. Preprint, arXiv:1304.1319.

\bibitem{Dr1999} 
B. K. Driver [1999]
The Lie bracket of adapted vector fields on Wiener spaces, 
Appl. Math. Optim. 39 (2): 179--210.

\bibitem{Ey2010} 
G. L. Eyink [2010]
Stochastic least-action principle for the incompressible Navier--Stokes equation,
Physica D 239: 1236--1240.

\bibitem{Fl2011}
F. Flandoli, [2011] \textit{Random Perturbation of PDEs and Fluid Dynamic Models},
Saint Flour summer school lectures 2010, Lecture Notes in Mathematics
n. 2015, Springer, Berlin.

\bibitem{FlMaNe2014}
F. Flandoli, M. Maurelli and M. Neklyudov [2014]
Noise prevents infinite stretching of the passive field in a stochastic vector advection
equation, \textit{J. Math. Fluid Mech.} 16: 805--822.

\bibitem{Gl2010}  
Y. E. Gliklikh [2010] Solutions of Burgers, Reynolds, and Navier--Stokes equations via stochastic perturbations of inviscid flows, \textit{J. Nonlin. Math. Phys.} 17:sup1, 15--29.

\bibitem{Go2005} 
D. A. Gomes [2005] A variational formulation for the Navier--Stokes equation, 
\textit{Comm. Math. Phys.} 257: 227--234.

\bibitem{HoRa2012} 
S. Hochgerner and T. S. Ratiu [2012]
Geometry of non-holonomic diffusion,
arXiv:1204.6438v1 [math-ph].

\bibitem{Ho2011}
D. D. Holm [2011] 
{\it Geometric Mechanics I: Dynamics and Symmetry}, World Scientific: Imperial College Press, Singapore (2nd edition) ISBN 978-1-84816-195-5.

\bibitem{HoMaRa1998} 
D. D. Holm, J. E. Marsden and T. S. Ratiu [1998]
The Euler--Poincar\'e equations and semidirect products
with applications to continuum theories,
Adv. in Math., 137: 1-81,
http://xxx.lanl.gov/abs/chao-dyn/9801015.

\bibitem{HoMaRa2002}  
D. D. Holm, J. E. Marsden and T. S. Ratiu [2002]
The Euler--Poincar\'{e} equations in geophysical fluid dynamics,
\\In {\it Large-Scale Atmosphere-Ocean Dynamics 2:
Geometric Methods and Models}. Edited by J. Norbury \& I.
Roulstone, Cambridge University Press: Cambridge, pp.
251--299.

\bibitem{IkWa1981}
N. Ikeda and S. Watanabe [1981] \textit{Stochastic Differential Equations and Diffusion Processes}. North Holland Publ. Co., Amsterdam.

\bibitem{InFu1979} 
A. Inoue and T. Funaki [1979]
On a new derivation of the Navier--Stokes equation,
\textit{Commun. Math. Phys.} 65: 83--90.

\bibitem{KaSh1994}
T. Kazumi and I. Shigekawa [1994]
Differential calculus on a submanifold of an abstract Wiener space, I. Covariant derivative,
In \textit{Stochastic Analysis on Infinite Dimensional Spaces}, 
H Kunita and  H.H. Kuo (Ed), Pitman Research Notes in Mathematics Series, pp 117--140.

\bibitem{LaCa-Or2008} 
J. A. L\'azaro-Cam\'i and J. P. Ortega [2008]
Stochastic Hamiltonian dynamical systems,
\textit{Rep. Math. Phys.}, 61 (1): 65--122.

\bibitem{MaTiVa2003} 
A. J. Majda, I. Timofeyev and E. Vanden-Eijden [2003] 
Systematic Strategies for Stochastic Mode Reduction in Climate
\textit{J. Atmos. Sci.} 60:1705--1722.

\bibitem{MaRa1994}
J. E. Marsden  and T. S. Ratiu [1994], {\em Introduction to Mechanics and Symmetry}. 
Texts in Applied Mathematics, Vol. 75. New York: Springer.

\bibitem{MoTs1992} 
Moffatt, H. K. and Tsinober, A. [1992] Helicity in laminar and turbulent flow. {\it Annu. Rev. Fluid Mech.} {\bf 24}, 281--312.

\bibitem{Pa2007} 
E. Pardoux [2007]
\textit{Stochastic Partial Differential Equations},
Lectures given in Fudan University, Shanghai. Published by Marseille, France.

\bibitem{Pedlosky87} 
J. Pedlosky [1987]
\textit{Geophysical Fluid Dynamics}, 2nd Edition, Springer, New York.

\bibitem{Ra2000} 
D. L. Rapoport [2000] 
Stochastic differential geometry and the random integration of
the Navier--Stokes equations and the kinematic dynamo problem on smooth
compact manifolds and Euclidean space, 
Hadronic J. 23: 637--675.

\bibitem{Ra2002} 
D. L. Rapoport [2002] 
On the geometry of the random representations for viscous fluids
and a remarkable pure noise representation,
Rep. Math. Phys. 50 (2): 211--250.

\bibitem{Sc1988} 
K.-U. Schauml\"offel [1988] White noise in space and time and the cylindrical Wiener process, \textit{Stochastic Analysis and Applications}, 6:1, 81--89.

\bibitem{Sh1996} 
{Shiryaev, A. N.} [1996] \textit{Probability}, {Graduate Texts in Mathematics},{\bf 95}, 2nd edition, Translated from the first (1980) Russian edition by R. P. Boas. {Springer-Verlag, New York}.

\bibitem{St1966} 
R. L. Stratonovich [1966] A new representation for stochastic integrals and equations,
J. SIAM Control, 4 (2): 362--371.

\bibitem{AW83} 
A. Weinstein [1983] Hamiltonian structure for drift waves
and geostrophic flow, Phys. Fluids 26: 388--390.

\bibitem{ZePa1994} 
V. Zeitlin and R. A. Pasmanter [1994] On the differential
geometry approach to the geophysical flows,
Phys. Lett. A 189: 59--63.

\end{thebibliography}
\end{document}